\documentclass[bibtex,submission,copyright,creativecommons]{eptcs}
\usepackage{amssymb,nicefrac}
\usepackage{amsfonts}
\usepackage[nocompress]{cite}
\usepackage{graphicx}
\usepackage{amsmath}
\allowdisplaybreaks
\usepackage{amsthm}
\usepackage[varg]{txfonts}
\usepackage{stmaryrd}

\usepackage{framed}
\usepackage[dvipsnames]{xcolor}
\usepackage{xargs} 
\usepackage[normalem]{ulem}
\usepackage{tikzfig}
\usepackage{placeins}

\usepackage{bm}
\usepackage{braket}

\usepackage{blkarray}
\usepackage{cleveref}

\usepackage{thmtools}
\usepackage{thm-restate}

\usepackage{hyperref}
\usepackage{lscape}

\newtheorem{Th}{Theorem}[section]
\newtheorem{theorem}[Th]{Theorem}
\newtheorem{proposition}[Th]{Proposition} 
\newtheorem{lemma}[Th]{Lemma}
\newtheorem{corollary}[Th]{Corollary}
\newtheorem{definition}[Th]{Definition} 
\newtheorem{example}[Th]{Example}
\newtheorem{remark}[Th]{Remark}

\usepackage{mathtools} 
\usepackage{bigints}

\usepackage[acronym]{glossaries}

\pgfplotsset{compat=1.17}

\tikzstyle{env}=[copoint,regular polygon rotate=0,minimum width=0.2cm, fill=black]

\tikzstyle{dot}=[inner sep=0.7mm,minimum width=0pt,minimum height=0pt,fill=black,draw=black,shape=circle]

\tikzstyle{every picture}=[baseline=-0.25em]
\tikzstyle{dotpic}=[scale=0.5]
\tikzstyle{diredges}=[every to/.style={diredge}]
\tikzstyle{dot graph}=[shorten <=-0.1mm,shorten >=-0.1mm,scale=0.6]
\tikzstyle{plot point}=[circle,fill=black,minimum width=2mm,inner sep=0]

\tikzstyle{dtriangle}=[fill=yellow,draw=black,shape=isosceles triangle,shape border rotate=-90,isosceles triangle stretches=true,inner sep=0.8pt,minimum width=0.25cm,minimum height=2mm]
\tikzstyle{vtriang}=[fill=yellow,draw=black,shape=isosceles triangle,shape border rotate=180,isosceles triangle stretches=true,inner sep=0.8pt,minimum width=0.25cm,minimum height=2mm]
\tikzstyle{trigmc}=[fill=green,draw=black,shape=isosceles triangle,shape border rotate=90,isosceles triangle stretches=true,inner sep=0.8pt,minimum width=0.3cm,minimum height=2mm]
\tikzstyle{vrt}=[fill=yellow,draw=black,shape=isosceles triangle,shape border rotate=0,isosceles triangle stretches=true,inner sep=0.8pt,minimum width=0.25cm,minimum height=2mm]
\tikzstyle{H box}=[rectangle,fill=yellow,draw=black,xscale=0.8,yscale=0.8, inner sep=0.6pt]
\tikzstyle{H box}=[rectangle,fill=yellow,draw=black,xscale=1.5,yscale=1.5, inner sep=1.6pt]
\tikzstyle{gbox}=[gn_phase, rounded corners=0]
\tikzstyle{rbox}=[rn_phase, rounded corners=0]
\tikzstyle{zhbx}=[rectangle,fill=white,draw=black,xscale=1.0,yscale=1.0, inner sep=1.6pt]
\tikzstyle{newh}=[rectangle,fill=yellow,draw=black,xscale=1.5,yscale=1.5, inner sep=1.6pt]
\tikzstyle{triangle}=[fill=yellow,draw=black,shape=isosceles triangle,shape border rotate=90,isosceles triangle stretches=true,inner sep=0.8pt,minimum width=0.25cm,minimum height=2mm]
\tikzstyle{vtriangle}=[triangle, rotate=180]

\definecolor{zx_red}{RGB}{232, 165, 165}
\definecolor{zx_green}{RGB}{216, 248, 216}

\tikzstyle{gn}=[dot,inner sep=0pt,minimum width=2mm,fill=zx_green]
\tikzstyle{rn}=[gn,fill=zx_red]

\tikzstyle{gn_phase}=[minimum size=1.2em, font={\footnotesize\boldmath}, shape=rectangle, rounded corners=0.5em, inner sep=0.2em, outer sep=-0.2em, scale=0.8, draw=black, fill=zx_green]
\tikzstyle{rn_phase}=[gn_phase, fill=zx_red]

\tikzstyle{dbspider}=[fill=black,draw=black,scale=1,shape=isosceles triangle,shape border rotate=90,isosceles triangle stretches=true,inner sep=1pt,minimum width=0.4cm,minimum height=3mm]

\tikzstyle{braceedge}=[decorate,decoration={brace,amplitude=2mm,raise=-1mm}]

\tikzstyle{wire label}=[font=\tiny, auto]

\tikzstyle{box}=[fill=white]

\tikzstyle{cdiag}=[matrix of math nodes, row sep=3em, column sep=3em, text height=1.5ex, text depth=0.25ex,inner sep=0.5em]
\tikzstyle{arrow above}=[transform canvas={yshift=0.5ex}]
\tikzstyle{arrow below}=[transform canvas={yshift=-0.5ex}]

\newacronym{opug}{OPUG}{One-Parameter Unitary Group (OPUG)}

\makeatletter
\newcommand{\vast}{\bBigg@{6.5}}
\makeatother

\usepackage{parskip}
\usepackage{tabu}

\title{How~to~Sum~and~Exponentiate~Hamiltonians~in~ZXW~Calculus}
\author{Razin A. Shaikh \and Quanlong Wang \and Richie Yeung \and
\institute{Quantinuum, 17 Beaumont Street, Oxford OX1 2NA, United Kingdom}}
\def\titlerunning{How to Sum and Exponentiate Hamiltonians in ZXW Calculus}
\def\authorrunning{R. A. Shaikh, Q. Wang \& R. Yeung}
\def\event{QPL 2022}

\begin{document}
\date{\today}\maketitle
\begin{abstract}
This paper develops practical summation techniques in ZXW calculus to reason about quantum dynamics, such as unitary time evolution.
First we give a direct representation of a wide class of sums of linear operators, including arbitrary qubits Hamiltonians, in ZXW calculus. As an application, we demonstrate the linearity of the Schrödinger equation and give a diagrammatic representation of the Hamiltonian in Greene-Diniz et al~\cite{Gabriel2022}, which is the first paper that models carbon capture using quantum computing.
We then use the Cayley-Hamilton theorem to show in principle how to exponentiate arbitrary qubits Hamiltonians in ZXW calculus. Finally, we develop practical techniques and show how to do Taylor expansion and Trotterization diagrammatically for Hamiltonian simulation. This sets up the framework for using ZXW calculus to the problems in quantum chemistry and condensed matter physics.
\end{abstract}


\section{Introduction}

ZX calculus \cite{CoeckeDuncan} is a graphical representation for quantum circuits and linear
algebra in general: a diagram with $n$ inputs and $m$ outputs represents a
$2^n \times 2^m$ linear map. ZX calculus comes with a complete set of rewrite
rules such that two diagrams that represent the same linear map can always be
rewritten to one another. Using this philosophy of
\textit{computation via rewriting}, ZX calculus has been successfully applied to a wide
range of problems in quantum computing, including circuit compilation \cite{Cowtan_2020,debeaudrapbianwang,duncan2020graph, Kissinger_2020}, 
circuit equality validation \cite{PhysRevA.102.022406, Lemonnier2021HypergraphSL},
circuit simulation \cite{Kissinger2022ClassicalSO}, error correction \cite{carette_et_al:LIPIcs:2019:10999, Chancellor2016GraphicalSF}, natural language processing~\cite{coecke2020qnlpfoundations,lorenz2021qnlp} and quantum machine learning~\cite{Stollenwerk2022PQC,yeung2020diagrammatic,Zhao2021analyzingbarren}.

To solve a problem using ZX calculus, or other diagrammatic
calculii, we first need to synthesise the expressions involved into diagrams.
This can be done in an ad-hoc way or using general methods such as elementary matrices \cite{wang2021representing}.
For example: in Hamiltonian simulation problems \cite{qcchemsurvey2020} the Hamiltonian is typically given as
a sum of Pauli operators, and without an efficient representation of such Hamiltonians in
ZX calculus, there is simply no way to proceed with the problem.

As of this writing, there is no published work on how to efficiently combine
sums of ZX diagrams. Both \cite{jeandel2022addition} and \cite{qwangnormalformbit}
first require an inefficient
conversion step to an intermediate form before the diagrams can be summed
together. Furthermore, the resulting diagram is large and does not resemble the original
symbolic expression, so the advantage of using diagrammatic reasoning is diminished.

To obtain intuitive sums of diagrams, we introduce the framework of \emph{ZXW calculus}. Coecke and Kissinger~\cite{Coeckealeksmen} proposed the idea of developing a graphical calculus based on the interaction between GHZ and W states. Following up on that, Hadzihasanovic~\cite{amar,amar1} developed the ZW calculus. The ZW-calculus has found applications in linear optical quantum computing~\cite{hadzihasanovicDiagrammaticAxiomatisationFermionic2018, defeliceQuantumLinearOptics2022}, describing interactions in quantum field theory~\cite{shaikhFeynman2022}, and studying multi-partite entanglement~\cite{Coeckealeksmen}. In this paper, we combine the ZX and ZW calculii to create the \emph{ZXW calculus}. The W-spider of the ZXW calculus plays an important role in producing efficient and compact sums of diagrams. We define the ZXW calculus through a slight modification of the algebraic ZX calculus~\cite{wangalg2020}.

In this work, we modify and extend the notion of controlled states~\cite{jpvnormfmlics} to
give direct representations of controlled diagrams
for a wide class of matrices and show how to sum them within the framework of ZXW calculus. These representations, combined with the recently developed
techniques of differentiating arbitrary ZX diagrams~\cite{wang2022differentiating, jeandel2022addition},
allow us to
\emph{practically reason about analytical problems that were previously inaccessible to ZX calculus}. Moreover, the addition, exponentiation and differentiation operations interact nicely, allowing us to work with ZX diagrams as naturally as the traditional matrix notation, while still keeping the compact representation and rewriting advantages of the ZX calculus.

Specifically, these techniques allow us to reason about quantum dynamics and quantum chemistry in the ZXW calculus. Previous diagrammatic approaches~\cite{Gogioso2015Schrodinger,Gogioso2017Thesis,Gogioso2019Dynamics} to quantum dynamics have tackled the problem from an abstract categorical perspective. In this paper, we develop diagrammatic techniques to tackle problems of quantum chemistry and condensed matter physics in a concrete and practical way. We first devise a diagrammatic form for arbitrary Hamiltonians in the ZXW calculus. One of the main problems in quantum computational chemistry is that of Hamiltonian simulation: given a Hamiltonian, find an approximation to its unitary time evolution. The ideal evolution is given by $e^{-iHt/2}$, where $H$ is the Hamiltonian and $t$ is time. Using the diagrammatic form of Hamiltonian and the summation techniques developed in this paper, we show how to diagrammatically exponentiate Hamiltonians. This allows us to write the unitary time evolution graphically and apply the ZXW rewrite rules to extract a quantum circuit for Hamiltonian simulation.

\subsection*{Summary of results}
\begin{enumerate}
 \item \textbf{Representing sums of square matrices and arbitrary vectors}: We show how to represent a sum of any square matrices of size $2^m \times 2^m$ or any vectors of size $2^m $ in ZXW calculus  (Section~\ref{controlledingen} and Section~\ref{controlledrealise}). As an application, We formulate the Schr\"{o}dinger equation in
ZXW calculus and show the linearity of its solutions. (Section~\ref{Schrodingers})
    \item \textbf{Representing arbitrary Hamiltonians}: We show how to construct a diagram for any Hamiltonian defined on arbitrary number of qubits (Section~\ref{sec:Hamiltonians}). As an example, we express the Hamiltonian used in Greene-Diniz et. al.~\cite{Gabriel2022} (\autoref{carbonhamitonian}).
     \item \textbf{Representing Taylor expansion and Trotterization}: We show how Taylor expansion and Trotterization used for practical Hamiltonian exponentiation can be realised in ZXW calculus. (Section~\ref{sec:HamiltonianExponentiation})
\end{enumerate}

\section{ZXW Calculus}

In this section, we give an introduction to  ZXW calculus which is a slight modification of the algebraic ZX calculus \cite{wangalg2020}, including its generators and rewriting rules. Note that algebraic ZX calculus is complete  \cite{wangalg2020} for matrices of size $2^m\times 2^n$ and hence, so is the ZXW calculus. 
In this paper diagrams are either read from left to right or top to bottom.

\subsection{Generators}
The diagrams in  ZXW calculus are defined by freely combining the following generating objects. Note that $a$ can be any complex number.

$$
\scalebox{0.9}{%
	\beginpgfgraphicnamed{zxwintr/generalgreenspider}
	\InputIfFileExists{zxwintr/generalgreenspider.tikz}{}{\input{./figures/zxwintr/generalgreenspider.tikz}}%
	\endpgfgraphicnamed
}
\qquad \quad
	\beginpgfgraphicnamed{zxwintr/newhadamard}
	\begin{tikzpicture}
	\begin{pgfonlayer}{nodelayer}
		\node [style=newh] (0) at (0, 0) {};
		\node [style=none] (1) at (0, 0.5) {};
		\node [style=none] (2) at (0, -0.5) {};
	\end{pgfonlayer}
	\begin{pgfonlayer}{edgelayer}
		\draw (1.center) to (2.center);
	\end{pgfonlayer}
\end{tikzpicture}}%
	\endpgfgraphicnamed

\qquad \quad
	\beginpgfgraphicnamed{zxwintr/w1to2}
	\begin{tikzpicture}
	\begin{pgfonlayer}{nodelayer}
		\node [style=none] (0) at (-0.25, -0.5) {};
		\node [style=none] (1) at (0, 0.5) {};
		\node [style=none] (2) at (0.25, -0.5) {};
		\node [style=dbspider] (3) at (0, 0) {};
	\end{pgfonlayer}
	\begin{pgfonlayer}{edgelayer}
		\draw (1.center) to (3);
		\draw (3) to (0.center);
		\draw (3) to (2.center);
	\end{pgfonlayer}
\end{tikzpicture}}%
	\endpgfgraphicnamed

\qquad \quad
	\beginpgfgraphicnamed{zxwintr/swap}
	\InputIfFileExists{zxwintr/swap.tikz}{}{\input{./figures/zxwintr/swap.tikz}}%
	\endpgfgraphicnamed

\qquad \quad
	\beginpgfgraphicnamed{zxwintr/Id}
	\begin{tikzpicture}
	\begin{pgfonlayer}{nodelayer}
		\node [style=none] (1) at (0.5, 0.3) {};
		\node [style=none] (2) at (0.5, -0.3) {};
		\node [style=none] (3) at (0.5, -0.5) {};
		\node [style=none] (4) at (0.5, 0.5) {};
	\end{pgfonlayer}
	\begin{pgfonlayer}{edgelayer}
		\draw (1.center) to (2.center);
	\end{pgfonlayer}
\end{tikzpicture}}%
	\endpgfgraphicnamed

\qquad \quad
	\beginpgfgraphicnamed{zxwintr/cap1}
	\begin{tikzpicture}
	\begin{pgfonlayer}{nodelayer}
		\node [style=none] (0) at (-0.5, -0.25) {};
		\node [style=none] (1) at (0.5, -0.25) {};
		\node [style=none] (2) at (0, 0.5) {};
	\end{pgfonlayer}
	\begin{pgfonlayer}{edgelayer}
		\draw [bend left=90, looseness=1.50] (0.center) to (1.center);
	\end{pgfonlayer}
\end{tikzpicture}
}%
	\endpgfgraphicnamed

\qquad \quad
	\beginpgfgraphicnamed{zxwintr/cup1}
	\begin{tikzpicture}
	\begin{pgfonlayer}{nodelayer}
		\node [style=none] (0) at (-0.5, 0.25) {};
		\node [style=none] (1) at (0.5, 0.25) {};
		\node [style=none] (2) at (0, 0.5) {};
	\end{pgfonlayer}
	\begin{pgfonlayer}{edgelayer}
		\draw [bend right=90, looseness=1.50] (0.center) to (1.center);
	\end{pgfonlayer}
\end{tikzpicture}
}%
	\endpgfgraphicnamed

$$

\subsection{Additional notation}
For simplicity, we introduce additional notation based on the given generators:

1. The green spider from the original ZX calculus can be defined using the green box spider in  ZXW calculus.
\begin{equation*}
    \scalebox{0.9}{%
	\beginpgfgraphicnamed{zxwintr/convenientdenote1}
	\InputIfFileExists{zxwintr/convenientdenote1.tikz}{}{\input{./figures/zxwintr/convenientdenote1.tikz}}%
	\endpgfgraphicnamed
}
\end{equation*}
2. The triangle and the inverse triangle can be expressed as follows. The transposes of the triangle and the inverse triangle can be drawn as inverted triangles.
\begin{equation*}
    \scalebox{0.9}{
	\beginpgfgraphicnamed{zxwintr/convenienttriangles}
	\InputIfFileExists{zxwintr/convenienttriangles.tikz}{}{\input{./figures/zxwintr/convenienttriangles.tikz}}%
	\endpgfgraphicnamed
\hspace{1cm}
	\beginpgfgraphicnamed{zxwintr/convenientdenote3}
	\InputIfFileExists{zxwintr/convenientdenote3.tikz}{}{\input{./figures/zxwintr/convenientdenote3.tikz}}%
	\endpgfgraphicnamed
}
\end{equation*}
3. The red spider from the original ZX calculus can be defined by performing Hadamard conjugation on each leg of the green spider, and the pink spider is the algebraic equivalent of the red spider. It is only defined for $\tau \in \{ 0, \pi \}$, and is rescaled to have integer components in its matrix representation.
\begin{equation*}
    \scalebox{0.9}{%
	\beginpgfgraphicnamed{zxwintr/convenientdenote5v2}
	\InputIfFileExists{zxwintr/convenientdenote5v2.tikz}{}{\input{./figures/zxwintr/convenientdenote5v2.tikz}}%
	\endpgfgraphicnamed
}
\end{equation*}
4. The plus gate (also known as $V$ gate) and minus gate ( also known as $V^{\dagger}$ gate) can be used for constructing the Pauli $Y$ gate, they are defined as  $V = HSH$ and $V^{\dagger} = HS^{\dagger}H$ respectively:

\begin{equation*}
    \scalebox{0.9}{%
	\beginpgfgraphicnamed{zxwintr/plusandminus}
	\InputIfFileExists{zxwintr/plusandminus.tikz}{}{\input{./figures/zxwintr/plusandminus.tikz}}%
	\endpgfgraphicnamed
}
\end{equation*}
5. The ``and'' box is defined as the following.
\begin{equation*}
    \scalebox{0.9}{%
	\beginpgfgraphicnamed{zxwintr/anddef}
	\InputIfFileExists{zxwintr/anddef.tikz}{}{\input{./figures/zxwintr/anddef.tikz}}%
	\endpgfgraphicnamed
}
\end{equation*}

\subsection{Rules}
Now we give the rewriting rules of  ZXW calculus.
\[
\begin{array}{|cccc|}
\hline
	\beginpgfgraphicnamed{zxwintr/generalgreenspiderfusesym}
	\InputIfFileExists{zxwintr/generalgreenspiderfusesym.tikz}{}{\input{./figures/zxwintr/generalgreenspiderfusesym.tikz}}%
	\endpgfgraphicnamed
&(S1) &%
	\beginpgfgraphicnamed{zxwintr/s2new2}
	\InputIfFileExists{zxwintr/s2new2.tikz}{}{\input{./figures/zxwintr/s2new2.tikz}}%
	\endpgfgraphicnamed
 &(S2)\\
	\beginpgfgraphicnamed{zxwintr/induced_compact_structure}
	\InputIfFileExists{zxwintr/induced_compact_structure.tikz}{}{\input{./figures/zxwintr/induced_compact_structure.tikz}}%
	\endpgfgraphicnamed
&(S3) &%
	\beginpgfgraphicnamed{zxwintr/rdotaempty}
	\InputIfFileExists{zxwintr/rdotaempty.tikz}{}{\input{./figures/zxwintr/rdotaempty.tikz}}%
	\endpgfgraphicnamed
  &(Ept) \\
	\beginpgfgraphicnamed{zxwintr/b1ring}
	\InputIfFileExists{zxwintr/b1ring.tikz}{}{\input{./figures/zxwintr/b1ring.tikz}}%
	\endpgfgraphicnamed
&(B1)  &%
	\beginpgfgraphicnamed{zxwintr/b2ring}
	\InputIfFileExists{zxwintr/b2ring.tikz}{}{\input{./figures/zxwintr/b2ring.tikz}}%
	\endpgfgraphicnamed
&(B2)\\ 
	\beginpgfgraphicnamed{zxwintr/rpicopyns}
	\InputIfFileExists{zxwintr/rpicopyns.tikz}{}{\input{./figures/zxwintr/rpicopyns.tikz}}%
	\endpgfgraphicnamed
 &(B3)&%
	\beginpgfgraphicnamed{zxwintr/anddflipns}
	\InputIfFileExists{zxwintr/anddflipns.tikz}{}{\input{./figures/zxwintr/anddflipns.tikz}}%
	\endpgfgraphicnamed
&(Brk) \\
 & &&\\
	\beginpgfgraphicnamed{zxwintr/triangleocopy}
	\InputIfFileExists{zxwintr/triangleocopy.tikz}{}{\input{./figures/zxwintr/triangleocopy.tikz}}%
	\endpgfgraphicnamed
 &(Bas0) &%
	\beginpgfgraphicnamed{zxwintr/trianglepicopyns}
	\InputIfFileExists{zxwintr/trianglepicopyns.tikz}{}{\input{./figures/zxwintr/trianglepicopyns.tikz}}%
	\endpgfgraphicnamed
&(Bas1)\\
	\beginpgfgraphicnamed{zxwintr/plus1}
	\InputIfFileExists{zxwintr/plus1.tikz}{}{\input{./figures/zxwintr/plus1.tikz}}%
	\endpgfgraphicnamed
&(Suc)&%
	\beginpgfgraphicnamed{zxwintr/triangleinvers}
	\InputIfFileExists{zxwintr/triangleinvers.tikz}{}{\input{./figures/zxwintr/triangleinvers.tikz}}%
	\endpgfgraphicnamed
  & (Inv) \\
   & &&\\
	\beginpgfgraphicnamed{zxwintr/zerotoredns}
	\InputIfFileExists{zxwintr/zerotoredns.tikz}{}{\input{./figures/zxwintr/zerotoredns.tikz}}%
	\endpgfgraphicnamed
&(Zero)&%
	\beginpgfgraphicnamed{zxwintr/eunoscalar2}
	\InputIfFileExists{zxwintr/eunoscalar2.tikz}{}{\input{./figures/zxwintr/eunoscalar2.tikz}}%
	\endpgfgraphicnamed
&(EU) \\
	\beginpgfgraphicnamed{zxwintr/lemma4}
	\InputIfFileExists{zxwintr/lemma4.tikz}{}{\input{./figures/zxwintr/lemma4.tikz}}%
	\endpgfgraphicnamed
&(Sym) & %
	\beginpgfgraphicnamed{zxwintr/associate}
	\InputIfFileExists{zxwintr/associate.tikz}{}{\input{./figures/zxwintr/associate.tikz}}%
	\endpgfgraphicnamed
 &(Aso)\\ 
 & &&\\
	\beginpgfgraphicnamed{zxwintr/TR1314combine2}
	\InputIfFileExists{zxwintr/TR1314combine2.tikz}{}{\input{./figures/zxwintr/TR1314combine2.tikz}}%
	\endpgfgraphicnamed
&(Pcy) &%
	\beginpgfgraphicnamed{zxwintr/wdecomp}
	\InputIfFileExists{zxwintr/wdecomp.tikz}{}{\input{./figures/zxwintr/wdecomp.tikz}}%
	\endpgfgraphicnamed
&(Wdc)\\
  		  		\hline
\end{array}\]  
where $a, b \in \mathbb C$. The vertically flipped versions of the rules are assumed to hold as well.

\subsection{Interpretation}
Although the generators in ZXW calculus are formal mathematical objects in their own right, in the context of this paper we interpret the generators as linear maps, so each ZXW diagram is equivalent to a vector or matrix.
\[
	\beginpgfgraphicnamed{zxwintr/generalgreenspider}
	\InputIfFileExists{zxwintr/generalgreenspider.tikz}{}{\input{./figures/zxwintr/generalgreenspider.tikz}}%
	\endpgfgraphicnamed
 =\ket{0}^{\otimes m}\bra{0}^{\otimes n}+a\ket{1}^{\otimes m}\bra{1}^{\otimes n}\qquad
	\beginpgfgraphicnamed{zxwintr/newhadamard}
	\begin{tikzpicture}
	\begin{pgfonlayer}{nodelayer}
		\node [style=newh] (0) at (0, 0) {};
		\node [style=none] (1) at (0, 0.5) {};
		\node [style=none] (2) at (0, -0.5) {};
	\end{pgfonlayer}
	\begin{pgfonlayer}{edgelayer}
		\draw (1.center) to (2.center);
	\end{pgfonlayer}
\end{tikzpicture}}%
	\endpgfgraphicnamed
=\frac{1}{\sqrt{2}}\begin{pmatrix}
       1 & 1 \\
       1 & -1
\end{pmatrix}
\qquad
	\beginpgfgraphicnamed{zxwintr/w1to2}
	\begin{tikzpicture}
	\begin{pgfonlayer}{nodelayer}
		\node [style=none] (0) at (-0.25, -0.5) {};
		\node [style=none] (1) at (0, 0.5) {};
		\node [style=none] (2) at (0.25, -0.5) {};
		\node [style=dbspider] (3) at (0, 0) {};
	\end{pgfonlayer}
	\begin{pgfonlayer}{edgelayer}
		\draw (1.center) to (3);
		\draw (3) to (0.center);
		\draw (3) to (2.center);
	\end{pgfonlayer}
\end{tikzpicture}}%
	\endpgfgraphicnamed
=\begin{pmatrix}
        1 & 0 \\
        0 & 1 \\
        0 & 1 \\
        0 & 0
 \end{pmatrix}
\]
\[
	\beginpgfgraphicnamed{zxwintr/Id}
	\begin{tikzpicture}
	\begin{pgfonlayer}{nodelayer}
		\node [style=none] (1) at (0.5, 0.3) {};
		\node [style=none] (2) at (0.5, -0.3) {};
		\node [style=none] (3) at (0.5, -0.5) {};
		\node [style=none] (4) at (0.5, 0.5) {};
	\end{pgfonlayer}
	\begin{pgfonlayer}{edgelayer}
		\draw (1.center) to (2.center);
	\end{pgfonlayer}
\end{tikzpicture}}%
	\endpgfgraphicnamed
=\begin{pmatrix}
       1 & 0 \\
       0 & 1
\end{pmatrix}\qquad
	\beginpgfgraphicnamed{zxwintr/swap}
	\InputIfFileExists{zxwintr/swap.tikz}{}{\input{./figures/zxwintr/swap.tikz}}%
	\endpgfgraphicnamed
=\begin{pmatrix}
        1 & 0 & 0 & 0 \\
        0 & 0 & 1 & 0 \\
        0 & 1 & 0 & 0 \\
        0 & 0 & 0 & 1 
 \end{pmatrix}\qquad
	\beginpgfgraphicnamed{zxwintr/cap1}
	\begin{tikzpicture}
	\begin{pgfonlayer}{nodelayer}
		\node [style=none] (0) at (-0.5, -0.25) {};
		\node [style=none] (1) at (0.5, -0.25) {};
		\node [style=none] (2) at (0, 0.5) {};
	\end{pgfonlayer}
	\begin{pgfonlayer}{edgelayer}
		\draw [bend left=90, looseness=1.50] (0.center) to (1.center);
	\end{pgfonlayer}
\end{tikzpicture}
}%
	\endpgfgraphicnamed
=\begin{pmatrix}
        1  \\
        0  \\
        0  \\
        1  \\
 \end{pmatrix}\qquad
	\beginpgfgraphicnamed{zxwintr/cup1}
	\begin{tikzpicture}
	\begin{pgfonlayer}{nodelayer}
		\node [style=none] (0) at (-0.5, 0.25) {};
		\node [style=none] (1) at (0.5, 0.25) {};
		\node [style=none] (2) at (0, 0.5) {};
	\end{pgfonlayer}
	\begin{pgfonlayer}{edgelayer}
		\draw [bend right=90, looseness=1.50] (0.center) to (1.center);
	\end{pgfonlayer}
\end{tikzpicture}
}%
	\endpgfgraphicnamed
=\begin{pmatrix}
        1 & 0 & 0 & 1 
         \end{pmatrix},
\]

\begin{remark}
Due to the associative rule (Aso),  we can define the $W$ spider
and give its interpretation as follows  \cite{qwangqditzw}:
\[%
	\beginpgfgraphicnamed{zxwintr/mlegsblackspider}
	\InputIfFileExists{zxwintr/mlegsblackspider.tikz}{}{\input{./figures/zxwintr/mlegsblackspider.tikz}}%
	\endpgfgraphicnamed
 \qquad
	\beginpgfgraphicnamed{zxwintr/mlegsblackspiderone}
	\InputIfFileExists{zxwintr/mlegsblackspiderone.tikz}{}{\input{./figures/zxwintr/mlegsblackspiderone.tikz}}%
	\endpgfgraphicnamed
=\underbrace{\ket{0\cdots0}}_{m}\bra{0}+\sum_{k=1}^{m}\overbrace{\ket{\underbrace{0\cdots 0}_{k-1} 1 0\cdots 0}}^{m}\bra{1}.
\]
As a consequence, we have
\begin{equation}\label{mlegsblackspider0and1eq}
	\beginpgfgraphicnamed{zxwintr/mlegsblackspider0and1}
	\InputIfFileExists{zxwintr/mlegsblackspider0and1.tikz}{}{\input{./figures/zxwintr/mlegsblackspider0and1.tikz}}%
	\endpgfgraphicnamed

\end{equation}
\end{remark}

\section{Controlled diagrams}\label{controlledingen}
We start by stating the definitions of controlled states and matrices, and how
to perform operations on them. Note that our definition of controlled states are slightly different from that of Jeandel et al \cite{jeandel2022addition}, as we send $\ket{0}$  to $\ket{0}^{\otimes n}$
instead of $\ket{+}^{\otimes n}$. 
They also have a notion of controlled matrices
by applying the map-state duality to the controlled state,  which maps $\ket{0}$ to $\ket{+}^{\otimes n} \bra{+}^{\otimes n}$ instead of the identity matrix as used in our definition, so their definition is not equivalent to ours.
\begin{definition} [Controlled matrix]
    The controlled matrix $\widetilde{M}$ corresponding to the matrix $M$ is a diagram
    \begin{equation*}
        \scalebox{0.9}{%
	\beginpgfgraphicnamed{controlled/ControlledMatrixDefinition0}
	\InputIfFileExists{controlled/ControlledMatrixDefinition0.tikz}{}{\input{./figures/controlled/ControlledMatrixDefinition0.tikz}}%
	\endpgfgraphicnamed
}
        \qquad
        \text{such that}
        \qquad
        \scalebox{0.9}{%
	\beginpgfgraphicnamed{controlled/ControlledMatrixDefinition1}
	\InputIfFileExists{controlled/ControlledMatrixDefinition1.tikz}{}{\input{./figures/controlled/ControlledMatrixDefinition1.tikz}}%
	\endpgfgraphicnamed
}
        \qquad
        \text{and}
        \qquad
        \scalebox{0.9}{%
	\beginpgfgraphicnamed{controlled/ControlledMatrixDefinition2}
	\InputIfFileExists{controlled/ControlledMatrixDefinition2.tikz}{}{\input{./figures/controlled/ControlledMatrixDefinition2.tikz}}%
	\endpgfgraphicnamed
}
    \end{equation*}
\end{definition}

\begin{definition} [Controlled state]
    The controlled state $\widetilde{V}$ corresponding to the state $V$ is a diagram
    \begin{equation*}
        \scalebox{0.9}{%
	\beginpgfgraphicnamed{controlled/ControlledStateDefinition0}
	\InputIfFileExists{controlled/ControlledStateDefinition0.tikz}{}{\input{./figures/controlled/ControlledStateDefinition0.tikz}}%
	\endpgfgraphicnamed
}
        \qquad
        \text{such that}
        \qquad
        \scalebox{0.9}{%
	\beginpgfgraphicnamed{controlled/ControlledStateDefinition1}
	\InputIfFileExists{controlled/ControlledStateDefinition1.tikz}{}{\input{./figures/controlled/ControlledStateDefinition1.tikz}}%
	\endpgfgraphicnamed
}
        \qquad
        \text{and}
        \qquad
        \scalebox{0.9}{%
	\beginpgfgraphicnamed{controlled/ControlledStateDefinition2}
	\InputIfFileExists{controlled/ControlledStateDefinition2.tikz}{}{\input{./figures/controlled/ControlledStateDefinition2.tikz}}%
	\endpgfgraphicnamed
}
    \end{equation*}
\end{definition}

Now, we show how to construct control diagrams for sums and products of matrices, given the controlled diagrams for the original matrices.
\begin{proposition}[Controlled product of matrices]\label{ProductControlledMatrix}
    Given controlled matrices $\widetilde{M_1}, \ldots, \widetilde{M_k}$ corresponding to matrices $M_1, \ldots, M_k$, the controlled matrix for $\prod_i M_i$ is given by
    \begin{equation}
        \scalebox{0.9}{%
	\beginpgfgraphicnamed{controlled/ControlledMatrixProduct}
	\InputIfFileExists{controlled/ControlledMatrixProduct.tikz}{}{\input{./figures/controlled/ControlledMatrixProduct.tikz}}%
	\endpgfgraphicnamed
}
    \end{equation}
\end{proposition}
\begin{proposition}[Controlled sum of matrices]\label{SumControlledMatrix}
    Given controlled matrices $\widetilde{M_1}, \ldots, \widetilde{M_k}$ corresponding to matrices $M_1, \ldots, M_k$ and complex numbers $c_1, \ldots, c_k$, the controlled matrix for $\sum_i c_i M_i$ is given by
    \begin{equation}
        \scalebox{0.9}{%
	\beginpgfgraphicnamed{controlled/ControlledMatrixSum}
	\InputIfFileExists{controlled/ControlledMatrixSum.tikz}{}{\input{./figures/controlled/ControlledMatrixSum.tikz}}%
	\endpgfgraphicnamed
}
    \end{equation}
\end{proposition}
Both of these can be verified by simply plugging in the standard basis states.
Similarly, we can construct the controlled diagram for sums of states.

\begin{proposition}[Controlled sum of states]\label{SumControlledState}
    Given controlled states $\widetilde{V_1}, \ldots, \widetilde{V_k}$ corresponding to states $V_1, \ldots, V_k$ and complex numbers $c_1, \ldots, c_k$, the controlled state for $\sum_i c_i V_i$ is given by
    \begin{equation}
        \scalebox{0.9}{%
	\beginpgfgraphicnamed{controlled/ControlledStateSum}
	\InputIfFileExists{controlled/ControlledStateSum.tikz}{}{\input{./figures/controlled/ControlledStateSum.tikz}}%
	\endpgfgraphicnamed
}
    \end{equation}
\end{proposition}

\section{Realising controlled diagrams}\label{controlledrealise}
In this section, we show how to directly construct the controlled matrix $\tilde{M}$ given a square matrix $M$, and the controlled state $\tilde{\psi}$ given a state $\psi$.

\subsection{Controlled matrix}
Consider a matrix $M$ of size $2^m \times 2^m$. As given in \cite{wang2021representing}, $M$ can be represented in diagrams as follows:
 \[  %
	\beginpgfgraphicnamed{candgate/squarematdecom}
	\InputIfFileExists{candgate/squarematdecom.tikz}{}{\input{./figures/candgate/squarematdecom.tikz}}%
	\endpgfgraphicnamed
 \]
 where $E_i, 1\leq i\leq k$ is an elementary matrix. As a consequence of \autoref{ProductControlledMatrix}, if we know how to construct any controlled elementary matrix, then we are able to depict the controlled matrix $M$. 
 It has been shown in \cite{wang2021representing} that the elementary matrix $E_i$ must be of one of the following forms:

\begin{tabular*}{\textwidth}{c @{\extracolsep{\fill}} cc}
    \scalebox{0.8}{%
	\beginpgfgraphicnamed{candgate/rowitimesa}
	\InputIfFileExists{candgate/rowitimesa.tikz}{}{\input{./figures/candgate/rowitimesa.tikz}}%
	\endpgfgraphicnamed
}
    &\scalebox{0.8}{%
	\beginpgfgraphicnamed{candgate/rowitimeaplusj}
	\InputIfFileExists{candgate/rowitimeaplusj.tikz}{}{\input{./figures/candgate/rowitimeaplusj.tikz}}%
	\endpgfgraphicnamed
}
    &\scalebox{0.8}{%
	\beginpgfgraphicnamed{candgate/rowitiswapj}
	\InputIfFileExists{candgate/rowitiswapj.tikz}{}{\input{./figures/candgate/rowitiswapj.tikz}}%
	\endpgfgraphicnamed
}\\[1.5em]
    row multiplication & row addition & row switching
\end{tabular*}

Then it can be verified by plugging standard basis that their corresponding controlled matrices can be obtained by simply adding a branch to the And-gate:
\begin{proposition}
    The controlled elementary matrices are given as:

    \begin{tabular*}{\textwidth}{c @{\extracolsep{\fill}} cc}
        \scalebox{0.8}{%
	\beginpgfgraphicnamed{candgate/rowitimesact}
	\InputIfFileExists{candgate/rowitimesact.tikz}{}{\input{./figures/candgate/rowitimesact.tikz}}%
	\endpgfgraphicnamed
}
        &\scalebox{0.8}{%
	\beginpgfgraphicnamed{candgate/rowitimeaplusjct}
	\InputIfFileExists{candgate/rowitimeaplusjct.tikz}{}{\input{./figures/candgate/rowitimeaplusjct.tikz}}%
	\endpgfgraphicnamed
}
        &\scalebox{0.8}{%
	\beginpgfgraphicnamed{candgate/rowitiswapjct}
	\InputIfFileExists{candgate/rowitiswapjct.tikz}{}{\input{./figures/candgate/rowitiswapjct.tikz}}%
	\endpgfgraphicnamed
}\\[1.5em]
        row multiplication & row addition & row switching
    \end{tabular*}
\end{proposition}

\subsection{Controlled state}
According to Wang~\cite{qwangnormalformbit, wangQufiniteZXcalculusUnified2022}, a state vector $\left(a_0, \ldots, a_{2^m-1} \right)^T$ of dimension $2^m$, $m \in \mathbb{N}$, can be represented in the following normal form:
\begin{equation}
	\beginpgfgraphicnamed{controlled/NormalForm2}
	\InputIfFileExists{controlled/NormalForm2.tikz}{}{\input{./figures/controlled/NormalForm2.tikz}}%
	\endpgfgraphicnamed

\end{equation}
Hence, we can realise any controlled state by constructing the controlled diagram of the above normal form. This is given in the following proposition.
\begin{proposition}
    Controlled state for a vector $\left(a_0, \ldots, a_{2^m-1} \right)^T$ is given by
    \begin{equation}
        \scalebox{0.9}{%
	\beginpgfgraphicnamed{controlled/ControlledNormalForm2}
	\InputIfFileExists{controlled/ControlledNormalForm2.tikz}{}{\input{./figures/controlled/ControlledNormalForm2.tikz}}%
	\endpgfgraphicnamed
}
    \end{equation}
\end{proposition}

Using the above result and \autoref{SumControlledState}, we can derive the expression for linear combination of states, represented in the normal form.
\begin{proposition}\label{2normalformssumpr}
   \[
    \scalebox{0.9}{%
	\beginpgfgraphicnamed{zxw/2normalformssumcr2}
	\InputIfFileExists{zxw/2normalformssumcr2.tikz}{}{\input{./figures/zxw/2normalformssumcr2.tikz}}%
	\endpgfgraphicnamed
}
    \]
\end{proposition}

\section{Schr\"odinger Equation}\label{Schrodingers}
As an application of the results from the previous sections, we prove the linearity of the solutions of the Schr\"odinger equation.
We recall the Schr\"odinger equation:
\begin{equation}\label{Schrodinger}
i \frac{\partial}{\partial t} \ket{\Psi(t)} = H\ket{\Psi(t)}
\end{equation}
where $t$ is time, $\ket{\Psi(t)}$ is the state vector of the quantum system in question, and $H$ is a Hamiltonian operator.
We can write the Schr\"odinger equation diagrammatically as:
\begin{equation}
    \scalebox{0.8}{%
	\beginpgfgraphicnamed{zxw/SchrodingerEquation}
	\InputIfFileExists{zxw/SchrodingerEquation.tikz}{}{\input{./figures/zxw/SchrodingerEquation.tikz}}%
	\endpgfgraphicnamed
}
\end{equation}
where $\ket{\Psi(t)}$ can be expressed using the normal form and $\ket{\Psi'(t)}$ can be obtained by applying the differentiation gadget~\cite{wang2022differentiating} to $\ket{\Psi(t)}$.
Along with this, we use \autoref{2normalformssumpr} to show that any linear combination of solutions of Schr\"odinger equation is also a solution.
\begin{proposition}\label{thm:schro-linear}
Assume that $\Psi(t)$ and $\Phi(t)$ satisfy the Schrödinger equation (\ref{Schrodinger}) and $a, b$ are arbitrary complex numbers, then so does  $a\Psi(t)+b\Phi(t)$, i.e.,
\begin{equation}
    \scalebox{0.8}{%
	\beginpgfgraphicnamed{zxw/sum2soulutionsabstract}
	\InputIfFileExists{zxw/sum2soulutionsabstract.tikz}{}{\input{./figures/zxw/sum2soulutionsabstract.tikz}}%
	\endpgfgraphicnamed
}
\end{equation}
\end{proposition}

\section{Representing Hamiltonians in ZXW}\label{sec:Hamiltonians}
Here we give an efficient representation for a wide class of matrices using the ZXW calculus.
\begin{lemma}\label{lem:summatrix}
Any matrix of the form $\sum^n_{i=1} \alpha_i c^{-1}_i\left(\bigotimes^m_{j=1} D(a_{ij})\right)c_i$,
where $D(a_{ij}) = \ket{0}\bra{0} + a_{ij}\ket{1}\bra{1}$,  $c_i$ is the conjugation, and $\alpha_i$ is a complex coefficient,
can be expressed in ZXW calculus as
\begin{center}
\scalebox{0.9}{%
	\beginpgfgraphicnamed{sum/lem0v2}
	\InputIfFileExists{sum/lem0v2.tikz}{}{\input{./figures/sum/lem0v2.tikz}}%
	\endpgfgraphicnamed
}
\end{center}
\end{lemma}

\begin{theorem}\label{lem:paulisum}
Any Hamiltonian $\sum^n_{i=1} \alpha_i \bigotimes^m_{j=1} P_{ij}$ can be expressed in ZXW calculus using controlled-Paulis.
\begin{center}\scalebox{0.9}{%
	\beginpgfgraphicnamed{sum/lem1v2}
	\InputIfFileExists{sum/lem1v2.tikz}{}{\input{./figures/sum/lem1v2.tikz}}%
	\endpgfgraphicnamed
}\end{center}
For each controlled-Pauli, there is a leg on the $j$-th qubit if $P_{ij} \neq I$, and $c_i$ is the Clifford conjugation corresponding to the Pauli operator $P_{ij} = c^\dagger_{ij} Z c_{ij}$.
\end{theorem}

\begin{example}\label{hamiexample1}
For the Hamiltonian $H = X_1X_2 + X_2X_3 - Z_1 - Z_2 - Z_3$, we have
\begin{equation*}
    \scalebox{0.9}{%
	\beginpgfgraphicnamed{sum/ex1}
	\InputIfFileExists{sum/ex1.tikz}{}{\input{./figures/sum/ex1.tikz}}%
	\endpgfgraphicnamed
}
\end{equation*}
\end{example}

For an even larger example, the Hamiltonian used in \cite{Gabriel2022} is shown in Figure~\ref{carbonhamitonian}.


\begin{proposition}\label{lem:paulicom}
The diagrammatic representation of controlled sum of Hamiltonians, thus the sum of Hamiltonians, in \autoref{lem:paulisum} respects commutativity of addition  (i.e. $\widetilde{P_i+P_j} = \widetilde{P_j+P_i}$):
$$
 \scalebox{0.9}{%
	\beginpgfgraphicnamed{sum/lem2}
	\InputIfFileExists{sum/lem2.tikz}{}{\input{./figures/sum/lem2.tikz}}%
	\endpgfgraphicnamed
}   
$$
\end{proposition}
While this proposition is obvious in non-diagrammatic calculations, our goal here is to diagrammatically characterise the commutative properties of controlled matrices.

\begin{landscape}
    \begin{figure}
        1. Hamiltonian truncated using
        \texttt{hamiltonian.compress(abs\_tol=0.019)}, 42 terms:
        
        \includegraphics[width=\linewidth]{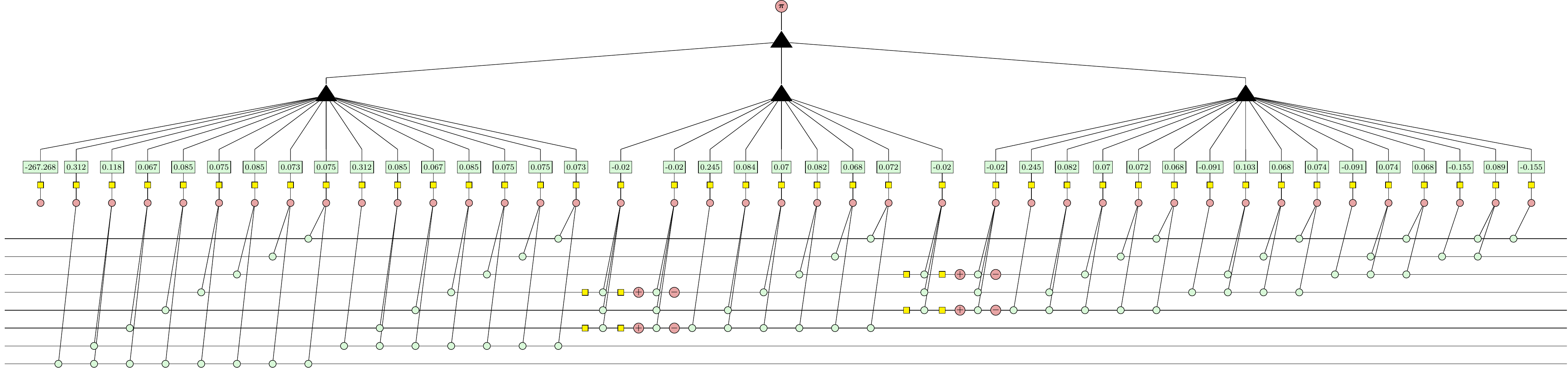}
        
        \vspace{4em}
        2. Hamiltonian truncated using
        \texttt{hamiltonian.compress(abs\_tol=0.008)}, 70 terms:
        
        \includegraphics[width=\linewidth]{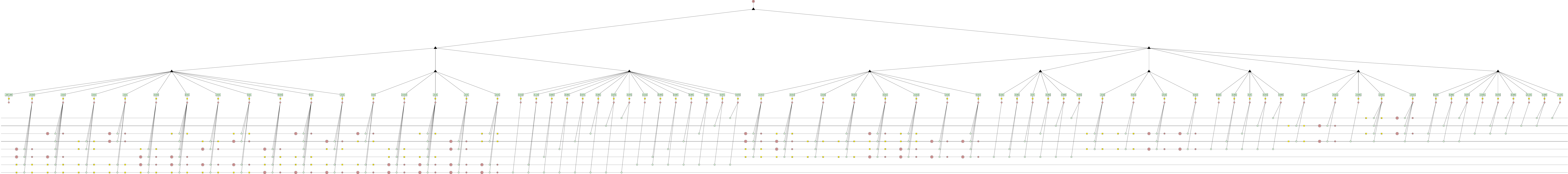}
        \caption{Hamiltonian in Greene-Diniz et al. \cite{Gabriel2022}}
        \label{carbonhamitonian}
    \end{figure}
\end{landscape}





\section{Hamiltonian Exponentiation}\label{sec:HamiltonianExponentiation}

According to Stone's theorem \cite{stone1932one}, the time-evolution operator of a Hamiltonian is the one-parameter unitary group (OPUG) generated by the Hamiltonian. Tasks such as Hamiltonian simulation require finding a unitary circuit that approximates the time evolution operator $e^{-iHt/2}$ for a given Hamiltonian~$H$. In general, finding such unitary circuit is a hard problem. Developing a method to obtain the one-parameter unitary group of Hamiltonian diagrams in ZXW will allow us tackle problems from quantum chemistry and condensed matter physics using diagrammatic tools. We could further use the rules of ZXW calculus to rewrite a Hamiltonian exponential diagram to a unitary time-evolution circuit.
In this section, we first talk about the exponentiation of Hamiltonians containing only commuting Pauli terms. We susequently describe the general case where the Hamiltonian may contain non-commuting terms.

\subsection{Hamiltonians with only commuting terms}\label{sec:HamiltonianCommuting}
For Hamiltonians comprised of only commuting Pauli terms, we have a simple correspondence between the diagrams of Hamiltonians and their exponentials. In the following example, we show how to obtain the Hamiltonian from its one-parameter unitary group.
We begin by writing the exponential of Pauli strings of the Hamiltonian using the phase gadgets~\cite{Cowtan_2020}. Then we differentiate the exponential diagram and set $t$ to 0: $e^{-iHt/2} \  \overset{\partial}{\mapsto} \  \frac{-i}{2} H e^{-iHt/2} \  \overset{t=0}{\mapsto} \  \frac{-i}{2} H$. As an example, consider the Hamiltonian $H = ZZZ + 2XZX$. We will omit the global phase in the diagrams for simplicity.

\ctikzfig{stone/stone2_proof}

We see from the above that the Hamiltonian diagrammatically commutes with its exponential.
We set $t=0$ to get the Hamiltonian:

\ctikzfig{stone/stone2}

In the following table, we show the correspondence between a few more Hamiltonians $H$ and their OPUGs $\Phi_H(t)$.
\begin{center}
    \bgroup
    \def\arraystretch{1.5}
    \begin{tabular}{|c|c|}
        \hline
        $Z \quad \leftrightarrow \quad \Phi_Z(t)$ &
        $ZXY \quad \leftrightarrow \quad \Phi_{ZXY}(t)$ \\ \hline
        \scalebox{0.9}{%
	\beginpgfgraphicnamed{stone/stone0}
	\begin{tikzpicture}[scale=0.5]
	\begin{pgfonlayer}{nodelayer}
		\node [style=none] (2) at (0, 2) {};
		\node [style=none] (3) at (3, 2) {};
		\node [style={gn_phase}] (6) at (1.5, 2) {$\pi$};
		\node [style=none] (15) at (5, 2) {};
		\node [style=none] (16) at (8, 2) {};
		\node [style={gn_phase}] (19) at (6.5, 2) {$t$};
		\node [style=none] (20) at (4, 2) {$\leftrightarrow$};
	\end{pgfonlayer}
	\begin{pgfonlayer}{edgelayer}
		\draw (2.center) to (3.center);
		\draw (15.center) to (16.center);
	\end{pgfonlayer}
\end{tikzpicture}
}%
	\endpgfgraphicnamed
} & \scalebox{0.9}{%
	\beginpgfgraphicnamed{stone/stone1}
	\begin{tikzpicture}[scale=0.5]
	\begin{pgfonlayer}{nodelayer}
		\node [style=none] (0) at (0, 0) {};
		\node [style=none] (1) at (0, 1) {};
		\node [style=none] (2) at (0, 2) {};
		\node [style=none] (3) at (3, 2) {};
		\node [style=none] (4) at (3, 1) {};
		\node [style=none] (5) at (3, 0) {};
		\node [style={gn_phase}] (6) at (1.5, 2) {$\pi$};
		\node [style={gn_phase}] (7) at (1.5, 1) {$\pi$};
		\node [style={gn_phase}] (8) at (1.5, 0) {$\pi$};
		\node [style=H box] (9) at (0.75, 1) {};
		\node [style=H box] (10) at (2.25, 1) {};
		\node [style=rn] (11) at (0.75, 0) {$\scriptstyle +$};
		\node [style=rn] (12) at (2.25, 0) {$\scriptstyle -$};
		\node [style=none] (13) at (5, 0) {};
		\node [style=none] (14) at (5, 1) {};
		\node [style=none] (15) at (5, 2) {};
		\node [style=none] (16) at (8, 2) {};
		\node [style=none] (17) at (8, 1) {};
		\node [style=none] (18) at (8, 0) {};
		\node [style=gn] (19) at (6.5, 2) {};
		\node [style=gn] (20) at (6.5, 1) {};
		\node [style=gn] (21) at (6.5, 0) {};
		\node [style=H box] (22) at (5.75, 1) {};
		\node [style=H box] (23) at (7.25, 1) {};
		\node [style=rn] (24) at (5.75, 0) {$\scriptstyle +$};
		\node [style=rn] (25) at (7.25, 0) {$\scriptstyle -$};
		\node [style=none] (26) at (4, 1) {$=$};
		\node [style=rn] (27) at (7.5, 3) {};
		\node [style={gn_phase}] (28) at (8.5, 3) {$\pi$};
		\node [style=none] (29) at (10, 0) {};
		\node [style=none] (30) at (10, 1) {};
		\node [style=none] (31) at (10, 2) {};
		\node [style=none] (32) at (13, 2) {};
		\node [style=none] (33) at (13, 1) {};
		\node [style=none] (34) at (13, 0) {};
		\node [style=gn] (35) at (11.5, 2) {};
		\node [style=gn] (36) at (11.5, 1) {};
		\node [style=gn] (37) at (11.5, 0) {};
		\node [style=H box] (38) at (10.75, 1) {};
		\node [style=H box] (39) at (12.25, 1) {};
		\node [style=rn] (40) at (10.75, 0) {$\scriptstyle +$};
		\node [style=rn] (41) at (12.25, 0) {$\scriptstyle -$};
		\node [style=none] (42) at (9, 1) {$\leftrightarrow$};
		\node [style=rn] (43) at (12.5, 3) {};
		\node [style={gn_phase}] (44) at (13.5, 3) {$t$};
		\node [style=none] (45) at (7.5, 3.75) {};
	\end{pgfonlayer}
	\begin{pgfonlayer}{edgelayer}
		\draw (2.center) to (3.center);
		\draw (4.center) to (1.center);
		\draw (0.center) to (5.center);
		\draw (15.center) to (16.center);
		\draw (17.center) to (14.center);
		\draw (13.center) to (18.center);
		\draw (27) to (19);
		\draw (27) to (20);
		\draw (27) to (21);
		\draw (28) to (27);
		\draw (31.center) to (32.center);
		\draw (33.center) to (30.center);
		\draw (29.center) to (34.center);
		\draw (43) to (35);
		\draw (43) to (36);
		\draw (43) to (37);
		\draw (44) to (43);
	\end{pgfonlayer}
\end{tikzpicture}
}%
	\endpgfgraphicnamed
} \\ \hline
        \multicolumn{2}{|c|}{
        $3 (XZY) - (ZZX) \quad \leftrightarrow \quad \Phi_{XZY}(3t) \ \Phi_{ZZX}(-t)$}\\ \hline
        \multicolumn{2}{|c|}{
            \scalebox{0.9}{%
	\beginpgfgraphicnamed{stone/stone3}
	\begin{tikzpicture}[scale=0.5]
	\begin{pgfonlayer}{nodelayer}
		\node [style=H box] (23) at (4.5, 2) {};
		\node [style=none] (13) at (5, 0) {};
		\node [style=none] (14) at (5, 1) {};
		\node [style=none] (15) at (5, 2) {};
		\node [style=none] (16) at (8.25, 2) {};
		\node [style=none] (17) at (8.25, 1) {};
		\node [style=none] (18) at (8.25, 0) {};
		\node [style=gn] (19) at (6.5, 2) {};
		\node [style=gn] (20) at (6.5, 1) {};
		\node [style=gn] (21) at (6.5, 0) {};
		\node [style=H box] (22) at (2.75, 2) {};
		\node [style=rn] (27) at (7.5, 3) {};
		\node [style=none] (42) at (9.25, 1) {$\leftrightarrow$};
		\node [style=none] (45) at (2, 0) {};
		\node [style=none] (46) at (2, 1) {};
		\node [style=none] (47) at (2, 2) {};
		\node [style=none] (48) at (5, 2) {};
		\node [style=none] (49) at (5, 1) {};
		\node [style=none] (50) at (5, 0) {};
		\node [style=gn] (51) at (3.5, 2) {};
		\node [style=gn] (52) at (3.5, 1) {};
		\node [style=gn] (53) at (3.5, 0) {};
		\node [style=rn] (58) at (4.5, 3) {};
		\node [style=H box] (60) at (5.75, 0) {};
		\node [style=H box] (61) at (7.5, 0) {};
		\node [style=H box] (62) at (12.75, 2) {};
		\node [style=none] (63) at (13.25, 0) {};
		\node [style=none] (64) at (13.25, 1) {};
		\node [style=none] (65) at (13.25, 2) {};
		\node [style=none] (66) at (16.5, 2) {};
		\node [style=none] (67) at (16.5, 1) {};
		\node [style=none] (68) at (16.5, 0) {};
		\node [style=gn] (69) at (14.75, 2) {};
		\node [style=gn] (70) at (14.75, 1) {};
		\node [style=gn] (71) at (14.75, 0) {};
		\node [style=H box] (72) at (11, 2) {};
		\node [style=rn] (73) at (15.75, 3) {};
		\node [style={gn_phase}] (74) at (16.75, 3) {$- t$};
		\node [style=none] (75) at (10.25, 0) {};
		\node [style=none] (76) at (10.25, 1) {};
		\node [style=none] (77) at (10.25, 2) {};
		\node [style=none] (78) at (13.25, 2) {};
		\node [style=none] (79) at (13.25, 1) {};
		\node [style=none] (80) at (13.25, 0) {};
		\node [style=gn] (81) at (11.75, 2) {};
		\node [style=gn] (82) at (11.75, 1) {};
		\node [style=gn] (83) at (11.75, 0) {};
		\node [style=rn] (86) at (12.75, 3) {};
		\node [style={gn_phase}] (87) at (13.75, 3) {$3 t$};
		\node [style=H box] (88) at (14, 0) {};
		\node [style=H box] (89) at (15.75, 0) {};
		\node [style=H box] (91) at (4.5, 3.75) {};
		\node [style=H box] (92) at (7.5, 3.75) {};
		\node [style=gbox] (93) at (4.5, 4.5) {$3$};
		\node [style=none] (94) at (7.5, 4.5) {};
		\node [style=dbspider] (95) at (6, 5.25) {};
		\node [style={rn_phase}] (96) at (6, 6.25) {$\pi$};
		\node [style=gbox] (97) at (7.5, 4.5) {$-1$};
		\node [style=rn] (98) at (2.75, 0) {$\scriptstyle +$};
		\node [style=rn] (99) at (4.25, 0) {$\scriptstyle -$};
		\node [style=rn] (100) at (11, 0) {$\scriptstyle +$};
		\node [style=rn] (101) at (12.5, 0) {$\scriptstyle -$};
		\node [style=none] (102) at (6.5, 7) {};
	\end{pgfonlayer}
	\begin{pgfonlayer}{edgelayer}
		\draw (15.center) to (16.center);
		\draw (17.center) to (14.center);
		\draw (13.center) to (18.center);
		\draw (27) to (19);
		\draw (27) to (20);
		\draw (47.center) to (48.center);
		\draw (49.center) to (46.center);
		\draw (45.center) to (50.center);
		\draw (58) to (51);
		\draw (58) to (52);
		\draw (58) to (53);
		\draw (65.center) to (66.center);
		\draw (67.center) to (64.center);
		\draw (63.center) to (68.center);
		\draw (73) to (69);
		\draw (73) to (70);
		\draw (73) to (71);
		\draw (74) to (73);
		\draw (77.center) to (78.center);
		\draw (79.center) to (76.center);
		\draw (75.center) to (80.center);
		\draw (86) to (81);
		\draw (86) to (82);
		\draw (86) to (83);
		\draw (87) to (86);
		\draw (91) to (58);
		\draw (92) to (27);
		\draw (93) to (91);
		\draw (94.center) to (92);
		\draw [in=90, out=-165, looseness=0.75] (95) to (93);
		\draw [in=90, out=-15, looseness=0.75] (95) to (94.center);
		\draw (96) to (95);
		\draw (27) to (21);
	\end{pgfonlayer}
\end{tikzpicture}
}%
	\endpgfgraphicnamed
}}\\ \hline
    \end{tabular}
    \egroup
\end{center}
For a more detailed reference of Pauli/phase gadgets and how to differentiate them, see \cite{Cowtan_2020} and \cite{koch2022quantum} respectively.

\subsection{General case}
In most of the interesting examples, the Hamiltonian contains some non-commuting terms. Constructing an exact circuit for the OPUG of such Hamiltonians is difficult. We will use the Cayley-Hamilton theorem~\cite{greub1975linearalgebra, cessa2011differential} to construct an exact diagram for the OPUG of the Hamiltonian. But calculating the coefficients in such diagram is computationally hard. Hence, we will present approximate diagrams of the OPUGs via Taylor expansion and Trotterization, which we can use in practical applications.

\subsubsection{Exact exponential diagram}

To give an exact diagram for the exponential, we can use the Cayley-Hamilton theorem~\cite{greub1975linearalgebra, cessa2011differential}. For a $n \times n$ matrix of the Hamiltonian $H$, we can express its exponential as 
\begin{equation}
    e^{-iHt/2} = c_0(t) I + c_1(t) H + \cdots + c_{n-1}(t) H^{n-1}
\end{equation}
where $c_0, \ldots, c_{n-1}$ are some functions of $t$.
Using this equation, we present the exact form of the Hamiltonian exponential function in ZXW:
\begin{equation}
    \scalebox{0.9}{%
	\beginpgfgraphicnamed{exponentials/MatrixExponential}
	\InputIfFileExists{exponentials/MatrixExponential.tikz}{}{\input{./figures/exponentials/MatrixExponential.tikz}}%
	\endpgfgraphicnamed
}
\end{equation}
The coefficients $c_0, \ldots, c_{n-1}$ can be calculated using Putzer's algorithm~\cite{putzer1966exponential}, but its computational complexity is exponential in the number of qubits.

Now, we demonstrate the utility of the ZXW calculus by using its rules to rewrite a Hamiltonian exponential diagram to a quantum circuit.
\begin{example}
    Consider the following Hamiltonian:
    \begin{equation}
        H = aX + bZ
    \end{equation}
    for any complex $a$ and $b$.
    There exists $s_0(t)$ and $s_1(t)$ such that
    \begin{equation}\label{eq:ExtractionHamiltonianExponential}
        \scalebox{0.9}{%
	\beginpgfgraphicnamed{exponentials/ExtractionHamiltonianExponential}
	\InputIfFileExists{exponentials/ExtractionHamiltonianExponential.tikz}{}{\input{./figures/exponentials/ExtractionHamiltonianExponential.tikz}}%
	\endpgfgraphicnamed
}
    \end{equation}
    Using ZXW, we can rewrite it to the following circuit. The full simplification steps are shown in Appendix~\ref{sec:CircuitExtraction}.
    \begin{equation}
        \scalebox{0.9}{%
	\beginpgfgraphicnamed{exponentials/ExtractionHamiltonianCircuit}
	\InputIfFileExists{exponentials/ExtractionHamiltonianCircuit.tikz}{}{\input{./figures/exponentials/ExtractionHamiltonianCircuit.tikz}}%
	\endpgfgraphicnamed
}
    \end{equation}
\end{example}

\subsubsection{Taylor expansion}
We consider Taylor expansion of the exponential to give a simple approximation of the OPUG in terms of ZXW diagram.
The expansion is given as: 
\begin{equation}
    e^{-iHt/2} = \sum_{k=0}^{\infty} \left(\frac{-i\,t}{2}\right)^k \frac{H^k}{k!}
\end{equation}
To construct a ZXW diagram, we need to take a finite number of terms from the expansion. Suppose, we take the $n$-th order approximation of the exponential. We can write it as the following diagram.
\begin{equation}
    \scalebox{1}{%
	\beginpgfgraphicnamed{exponentials/MatrixExponentialTaylor}
	\InputIfFileExists{exponentials/MatrixExponentialTaylor.tikz}{}{\input{./figures/exponentials/MatrixExponentialTaylor.tikz}}%
	\endpgfgraphicnamed
}
\end{equation}

\begin{example}
    Consider the Hamiltonian $H = a Z_1 Z_2 + b Z_1 X_2$. The diagram for this Hamiltonian is
    \begin{equation*}
        \scalebox{0.9}{%
	\beginpgfgraphicnamed{exponentials/MatrixExponentialExampleHamiltonian}
	\InputIfFileExists{exponentials/MatrixExponentialExampleHamiltonian.tikz}{}{\input{./figures/exponentials/MatrixExponentialExampleHamiltonian.tikz}}%
	\endpgfgraphicnamed
}
    \end{equation*}
    The third order approximation of $e^{-iHt/2}$ is the following diagram
    \begin{equation*}
        \scalebox{0.9}{%
	\beginpgfgraphicnamed{exponentials/MatrixExponentialExample}
	\InputIfFileExists{exponentials/MatrixExponentialExample.tikz}{}{\input{./figures/exponentials/MatrixExponentialExample.tikz}}%
	\endpgfgraphicnamed
}
    \end{equation*}
\end{example}

\subsubsection{Trotterization}
Using the Trotter-Suzuki formula, we can approximate the time evolution operator of the Hamiltonian by evolving the system in small time-steps. With small enough time-steps, we can treat the non-commuting terms in the Hamiltonian as commuting terms. Hence, for a Hamiltonian $H = \sum_k H_k$, we get
\begin{equation}
    e^{-iHt/2} = \left(\prod_k e^{-iH_kt/2n}\right)^n + \mathcal{O}\left(\frac{t^2}{n}\right)
\end{equation}
where $n$ is the number of Trotter steps.
This allows us to apply the techniques developed in Section~\ref{sec:HamiltonianCommuting} to construct the diagram for the OPUG of Hamiltonian.
\begin{example}
    Consider the Hamiltonian $H = 3ZY + 2ZZ$ represented by the following diagram.
    \begin{equation*}
        \scalebox{0.9}{%
	\beginpgfgraphicnamed{exponentials/TrotterizationExampleHamiltonian}
	\InputIfFileExists{exponentials/TrotterizationExampleHamiltonian.tikz}{}{\input{./figures/exponentials/TrotterizationExampleHamiltonian.tikz}}%
	\endpgfgraphicnamed
}
    \end{equation*}
    By approximating $e^{-iHt/2}$ with 5 Trotter steps, we obtain the diagram shown below.
    \begin{equation*}
        \scalebox{0.9}{%
	\beginpgfgraphicnamed{exponentials/TrotterizationExample}
	\InputIfFileExists{exponentials/TrotterizationExample.tikz}{}{\input{./figures/exponentials/TrotterizationExample.tikz}}%
	\endpgfgraphicnamed
}
    \end{equation*}
\end{example}

\section{Future work}
In this paper, we give direct representations of controlled diagrams
for a wide class of matrices and show how to sum them. As applications, we express any Hamiltonian with a form of a sum of Pauli operators, including the Hamiltonians used for carbon capture in \cite{Gabriel2022},   convert between a Hamiltonian and its one-parameter unitary group, and represent Taylor expansion and Trotterizatio used for practical Hamiltonian exponentiation in ZXW calculus.

We would like to apply the summation techniques developed in this paper to practical problems, like quantum approximate
optimization, integration on arbitrary ZX diagrams, or quantum machine learning. Also, we would like to develop tools to rewrite Hamiltonian exponentiation diagrams to quantum circuits.

\section*{Acknowledgements}
We would like to thank Gabriel Greene-Diniz and David Zsolt Manrique for providing the Pauli operators form of the Hamiltonian from their paper \cite{Gabriel2022}.  We would also like to thank Bob Coecke, Pablo Andres-Martinez, Boldizsár Poór, Lia Yeh, Stefano Gogioso, Konstantinos Meichanetzidis, Vincent Wang and Robin Lorenz for the feedback on this paper.

\bibliographystyle{eptcs}
\bibliography{generic.bib}

\appendix

\section{Proofs and Lemmas}
In this appendix, we include all the lemmas which have been essentially existed (up to scalars) in previous papers. The lemmas are given in the order which they appear in this paper.

\subsection{Useful lemmas}
We present several useful results in the following table.

{\tabulinesep=1.2mm
\begin{tabu}{|c|c|}
\hline
\begin{minipage}[t]{0.45\linewidth}
  \begin{lemma}\cite{qwangnormalformbit}\\
    For $\tau, \sigma \in \{0, \pi\}$, 
    pink spiders fuse.\\
$$%
	\beginpgfgraphicnamed{zxwintr/redspider0pifusion}
	\InputIfFileExists{zxwintr/redspider0pifusion.tikz}{}{\input{./figures/zxwintr/redspider0pifusion.tikz}}%
	\endpgfgraphicnamed
 (S1r)  $$
    \end{lemma} 
\end{minipage}&
\begin{minipage}[t]{0.45\linewidth}
   \begin{lemma}\cite{qwangnormalformbit}\\
    Hadamard is involutive.
    \vspace{0.25cm}
    $$%
	\beginpgfgraphicnamed{zxwintr/nhsquare2}
	\InputIfFileExists{zxwintr/nhsquare2.tikz}{}{\input{./figures/zxwintr/nhsquare2.tikz}}%
	\endpgfgraphicnamed
 \text{  (H2)}$$
    \end{lemma}
    \vspace{0.20cm}
\end{minipage}\\\hline
\begin{minipage}[t]{0.45\linewidth}
    \begin{lemma}\cite{qwangnormalformbit}\label{zx2elm}\\
      Pink $\pi$ transposes the triangle.\\
      \vspace{0.25cm}
	\beginpgfgraphicnamed{zxwintr/zx2e}
	\InputIfFileExists{zxwintr/zx2e.tikz}{}{\input{./figures/zxwintr/zx2e.tikz}}%
	\endpgfgraphicnamed

      \vspace{0.10cm}
    \end{lemma}
\end{minipage}&
\begin{minipage}[t]{0.45\linewidth}
  \begin{lemma}\cite{qwangnormalformbit}\label{traingleinverse}\\
    Green $\pi$ inverts the triangle.\\
	\beginpgfgraphicnamed{zxwintr/definitionTriangleInverse2}
	\InputIfFileExists{zxwintr/definitionTriangleInverse2.tikz}{}{\input{./figures/zxwintr/definitionTriangleInverse2.tikz}}%
	\endpgfgraphicnamed
 
 \end{lemma}
\end{minipage}\\\hline
\begin{minipage}[t]{0.45\linewidth}
   \begin{lemma}\cite{qwangnormalformbit}\label{trianglerpidotlm}\\
      $(\text{triangle})^T$ stabilises $\ket{1}$.\\
	\beginpgfgraphicnamed{zxwintr/trianglerpidot}
	\InputIfFileExists{zxwintr/trianglerpidot.tikz}{}{\input{./figures/zxwintr/trianglerpidot.tikz}}%
	\endpgfgraphicnamed
 
    \end{lemma}
\end{minipage}&
\begin{minipage}[t]{0.45\linewidth}
  \begin{lemma}\cite{qwangnormalformbit}\label{hopfnslm}\\
      Hopf rule.\\
$$%
	\beginpgfgraphicnamed{zxwintr/hopfns}
	\InputIfFileExists{zxwintr/hopfns.tikz}{}{\input{./figures/zxwintr/hopfns.tikz}}%
	\endpgfgraphicnamed
 (Hopf)$$
\end{lemma}
\end{minipage}\\\hline
\begin{minipage}[t]{0.45\linewidth}
\begin{lemma}\cite{qwangnormalformbit}\label{pimultiplecplm}\\
  $\pi$ copy rule. For $m \geq 0$:\\
$$%
	\beginpgfgraphicnamed{zxwintr/pigrcopy}
	\InputIfFileExists{zxwintr/pigrcopy.tikz}{}{\input{./figures/zxwintr/pigrcopy.tikz}}%
	\endpgfgraphicnamed
 (Pic)$$
\end{lemma}
\end{minipage}&
\begin{minipage}[t]{0.45\linewidth}
\begin{lemma}\cite{qwangnormalformbit}\label{1iprf}\\
$\pi$ commutation rule.\\
	\beginpgfgraphicnamed{zxwintr/picommutationdm}
	\InputIfFileExists{zxwintr/picommutationdm.tikz}{}{\input{./figures/zxwintr/picommutationdm.tikz}}%
	\endpgfgraphicnamed
 
\vspace{0.30cm}
  \end{lemma}
\end{minipage}\\\hline
\end{tabu}}

\begin{lemma}
Suppose $a\neq 0, a\in  \mathbb C$. Then
\[ %
	\beginpgfgraphicnamed{zxw/2formcphasegen}
	\InputIfFileExists{zxw/2formcphasegen.tikz}{}{\input{./figures/zxw/2formcphasegen.tikz}}%
	\endpgfgraphicnamed
 \]
\end{lemma}
This equality can be verified by plugging in the standard basis on the inputs of the diagrams.
Next, we have
\begin{lemma}
  \[ %
	\beginpgfgraphicnamed{zxw/generalpsgte}
	\InputIfFileExists{zxw/generalpsgte.tikz}{}{\input{./figures/zxw/generalpsgte.tikz}}%
	\endpgfgraphicnamed
\]
\end{lemma}
As a consequence of the above two lemmas, we have 
\begin{lemma}
  \[ %
	\beginpgfgraphicnamed{zxw/phasedetriangle}
	\InputIfFileExists{zxw/phasedetriangle.tikz}{}{\input{./figures/zxw/phasedetriangle.tikz}}%
	\endpgfgraphicnamed
 \]
\end{lemma}

\begin{lemma}\label{wfuse1eq}
  \[  %
	\beginpgfgraphicnamed{zxw/wfuse1}
	\InputIfFileExists{zxw/wfuse1.tikz}{}{\input{./figures/zxw/wfuse1.tikz}}%
	\endpgfgraphicnamed
 \]
\end{lemma}
where $a_1, \cdots, a_s$ are arbitrary complex numbers.
\begin{proof}
  \[\scalebox{0.9}{%
	\beginpgfgraphicnamed{zxw/wfuse1_proof}
	\InputIfFileExists{zxw/wfuse1_proof.tikz}{}{\input{./figures/zxw/wfuse1_proof.tikz}}%
	\endpgfgraphicnamed
}\]
\end{proof}

\begin{lemma}\label{psinormalformdifflm}
  \[   \frac{\partial}{\partial t}\left[%
	\beginpgfgraphicnamed{zxw/psinormalform2}
	\InputIfFileExists{zxw/psinormalform2.tikz}{}{\input{./figures/zxw/psinormalform2.tikz}}%
	\endpgfgraphicnamed
\right] = %
	\beginpgfgraphicnamed{zxw/psinormalformdiff2}
	\InputIfFileExists{zxw/psinormalformdiff2.tikz}{}{\input{./figures/zxw/psinormalformdiff2.tikz}}%
	\endpgfgraphicnamed
 \]
\end{lemma}
\begin{proof}
  \[   \frac{\partial}{\partial t}\left[%
	\beginpgfgraphicnamed{zxw/psinormalform2}
	\InputIfFileExists{zxw/psinormalform2.tikz}{}{\input{./figures/zxw/psinormalform2.tikz}}%
	\endpgfgraphicnamed
\right] = %
	\beginpgfgraphicnamed{zxw/psinormalformdiffprf2}
	\InputIfFileExists{zxw/psinormalformdiffprf2.tikz}{}{\input{./figures/zxw/psinormalformdiffprf2.tikz}}%
	\endpgfgraphicnamed
 \]
  \[= %
	\beginpgfgraphicnamed{zxw/psinormalformdiff2}
	\InputIfFileExists{zxw/psinormalformdiff2.tikz}{}{\input{./figures/zxw/psinormalformdiff2.tikz}}%
	\endpgfgraphicnamed
 \]
\end{proof}

\begin{lemma}\label{paulianticommuteeq}
  \[ %
	\beginpgfgraphicnamed{sum/lem2_1}
	\InputIfFileExists{sum/lem2_1.tikz}{}{\input{./figures/sum/lem2_1.tikz}}%
	\endpgfgraphicnamed
 \]
\end{lemma}
This shows that a Hadamard edge is
added to the `body' of the Pauli gadgets when their Hamiltonians anti-commute,
the proof can be found in \cite[Theorem 3]{yeung2020diagrammatic}.

\subsection{Proofs}

\begin{proof}[Proof of \autoref{thm:schro-linear}]
    \[\mkern-24mu %
	\beginpgfgraphicnamed{zxw/sum2soulutionsprf0}
	\InputIfFileExists{zxw/sum2soulutionsprf0.tikz}{}{\input{./figures/zxw/sum2soulutionsprf0.tikz}}%
	\endpgfgraphicnamed
\]
    \[%
	\beginpgfgraphicnamed{zxw/sum2soulutionsprf1v2}
	\InputIfFileExists{zxw/sum2soulutionsprf1v2.tikz}{}{\input{./figures/zxw/sum2soulutionsprf1v2.tikz}}%
	\endpgfgraphicnamed
\]
    \[%
	\beginpgfgraphicnamed{zxw/sum2soulutionsprf2v2}
	\InputIfFileExists{zxw/sum2soulutionsprf2v2.tikz}{}{\input{./figures/zxw/sum2soulutionsprf2v2.tikz}}%
	\endpgfgraphicnamed
\]
    \[%
	\beginpgfgraphicnamed{zxw/sum2soulutionsprf3v2}
	\InputIfFileExists{zxw/sum2soulutionsprf3v2.tikz}{}{\input{./figures/zxw/sum2soulutionsprf3v2.tikz}}%
	\endpgfgraphicnamed
\]
    \[%
	\beginpgfgraphicnamed{zxw/sum2soulutionsprf4v2}
	\InputIfFileExists{zxw/sum2soulutionsprf4v2.tikz}{}{\input{./figures/zxw/sum2soulutionsprf4v2.tikz}}%
	\endpgfgraphicnamed
\]
\end{proof}

\begin{proof}[Proof of \autoref{lem:summatrix}]
    We need check that the controlled matrix $\tilde{M_i}$ represents the controlled matrix of $\bigotimes^m_{j=1} D(a_{ij})$.
    After that, the rest of the proof follows from \autoref{SumControlledMatrix}.
    \ctikzfig{sum/lem3_pf2}
\end{proof}

\begin{proof}[Proof of \autoref{lem:paulisum}]
    In the \autoref{lem:summatrix}, we set
    \begin{equation*}
        a_{ij} = \begin{cases}
            e^{i0} & \text{if }\ P_{ij} = I\\
            e^{i\pi} & \text{otherwise.}
        \end{cases}
    \end{equation*}
    When $a_{ij} = e^{i0}$, the leg on the $j$-th qubit will be disconnected:
    \[ \scalebox{1}{%
	\beginpgfgraphicnamed{sum/thm5_2_1}
	\InputIfFileExists{sum/thm5_2_1.tikz}{}{\input{./figures/sum/thm5_2_1.tikz}}%
	\endpgfgraphicnamed
} \]
    On the other hand, for $a_{ij} = e^{i\pi}$, we get
    \[ \scalebox{1}{%
	\beginpgfgraphicnamed{sum/thm5_2_2}
	\InputIfFileExists{sum/thm5_2_2.tikz}{}{\input{./figures/sum/thm5_2_2.tikz}}%
	\endpgfgraphicnamed
} \]
    Substituting the above in the controlled matrix $\tilde{M_i}$, we get
    \[ \scalebox{1}{%
	\beginpgfgraphicnamed{sum/thm5_2_3}
	\InputIfFileExists{sum/thm5_2_3.tikz}{}{\input{./figures/sum/thm5_2_3.tikz}}%
	\endpgfgraphicnamed
} \]
    where there is a leg on the $j$-th qubit if $P_{ij} \neq I$.
\end{proof}

\begin{proof}[Proof of \autoref{lem:paulicom}]
    \[  %
	\beginpgfgraphicnamed{sum/lem2_pf1}
	\InputIfFileExists{sum/lem2_pf1.tikz}{}{\input{./figures/sum/lem2_pf1.tikz}}%
	\endpgfgraphicnamed
 \]
    Here for the second equality we have 
    \[  %
	\beginpgfgraphicnamed{sum/lem2_pf2}
	\InputIfFileExists{sum/lem2_pf2.tikz}{}{\input{./figures/sum/lem2_pf2.tikz}}%
	\endpgfgraphicnamed
 \]
\end{proof}

\section{Circuit extraction of the exponential from \autoref{eq:ExtractionHamiltonianExponential}} \label{sec:CircuitExtraction}
To simplify this diagram to a circuit, we will use the following two propositions.
\begin{proposition}\label{prop:WGreenEqualsCup}
    \[  %
	\beginpgfgraphicnamed{exponentials/WGreenEqualsCup}
	\InputIfFileExists{exponentials/WGreenEqualsCup.tikz}{}{\input{./figures/exponentials/WGreenEqualsCup.tikz}}%
	\endpgfgraphicnamed
 \]
\end{proposition}
\begin{proof}[Proof of \autoref{prop:WGreenEqualsCup}]
  \[  %
	\beginpgfgraphicnamed{exponentials/WGreenEqualsCupProof}
	\InputIfFileExists{exponentials/WGreenEqualsCupProof.tikz}{}{\input{./figures/exponentials/WGreenEqualsCupProof.tikz}}%
	\endpgfgraphicnamed
 \]
\end{proof}
\begin{proposition}\label{prop:TriangleGreenHad}
    \[  %
	\beginpgfgraphicnamed{exponentials/TriangleGreenHad}
	\InputIfFileExists{exponentials/TriangleGreenHad.tikz}{}{\input{./figures/exponentials/TriangleGreenHad.tikz}}%
	\endpgfgraphicnamed
 \]
\end{proposition}
\begin{proof}[Proof of \autoref{prop:TriangleGreenHad}]
  \[  %
	\beginpgfgraphicnamed{exponentials/TriangleGreenHadProof}
	\InputIfFileExists{exponentials/TriangleGreenHadProof.tikz}{}{\input{./figures/exponentials/TriangleGreenHadProof.tikz}}%
	\endpgfgraphicnamed
 \]
\end{proof}
Now, we begin simplifying~\eqref{eq:ExtractionHamiltonianExponential}.
\begin{align*}
    &%
	\beginpgfgraphicnamed{exponentials/ExtractionSimplify1}
	\InputIfFileExists{exponentials/ExtractionSimplify1.tikz}{}{\input{./figures/exponentials/ExtractionSimplify1.tikz}}%
	\endpgfgraphicnamed
\\[1em]
    &%
	\beginpgfgraphicnamed{exponentials/ExtractionSimplify2}
	\InputIfFileExists{exponentials/ExtractionSimplify2.tikz}{}{\input{./figures/exponentials/ExtractionSimplify2.tikz}}%
	\endpgfgraphicnamed
\\[1em]
    &%
	\beginpgfgraphicnamed{exponentials/ExtractionSimplify3}
	\InputIfFileExists{exponentials/ExtractionSimplify3.tikz}{}{\input{./figures/exponentials/ExtractionSimplify3.tikz}}%
	\endpgfgraphicnamed
\\[1em]
    &%
	\beginpgfgraphicnamed{exponentials/ExtractionSimplify4}
	\InputIfFileExists{exponentials/ExtractionSimplify4.tikz}{}{\input{./figures/exponentials/ExtractionSimplify4.tikz}}%
	\endpgfgraphicnamed

\end{align*}

\end{document}